\RequirePackage{fix-cm}
\documentclass{svjour3}                     
\smartqed  
\usepackage{graphicx}
\usepackage{pslatex}
\usepackage{amsfonts,color}
\usepackage{amssymb,amsmath,latexsym}

\makeatletter
\def\inmod#1{\allowbreak\mkern5mu({\operator@font mod}\,\,#1)}
\makeatother

\newcommand{\wt}{{\mathrm{wt}}}

\newcommand{\tr}{{\mathrm{Tr}}}
\newcommand{\gf}{{\mathrm{GF}}}

\newcommand{\C}{{\mathcal{C}}}

\newcommand{\bc}{{\mathbf{c}}}

\begin{document}

\title{Linear Codes from a Generic Construction
}


\author{Can Xiang}


\institute{Can Xiang \at
                College of Mathematics and Information Science, Guangzhou University, Guangzhou 510006, China \\
              \email{cxiangcxiang@hotmail.com}
}

\date{Received: date / Accepted: date}

\maketitle

\begin{abstract}
A generic construction of linear codes over finite fields has recently received a lot of attention, and many
one-weight, two-weight and three-weight codes with good error correcting capability have been produced
with this generic approach. The first objective of this paper is to  establish relationships among some classes
of linear codes obtained with this approach, so that the parameters of some classes of linear codes can be
derived from those of other classes with known parameters. In this way, linear codes with new parameters
will be derived.  The second is to present a class of three-weight binary codes and consider their applications
in secret sharing.

\end{abstract}

\section{Introduction}\label{sec-intro}

Throughout this paper, let $p$ be a prime and let $q=p^m$ for some positive integer $m$.
An $[n,\,\kappa,\,\omega]$ linear code $\C$ over $\gf(p)$ is a $\kappa$-dimensional subspace of $\gf(p)^n$ with minimum
(Hamming) distance $\omega$.
Let $A_i$ denote the number of codewords with Hamming weight $i$ in a code
$\C$ of length $n$. The {\em weight enumerator} of $\C$ is defined by
$
1+A_1z+A_2z^2+ \cdots + A_nz^n.
$
A code $\C$ is said to be a $t$-weight code  if the number of nonzero
$A_i$ in the sequence $(A_1, A_2, \cdots, A_n)$ is equal to $t$.The sequence
 $(1, A_1, A_2, \cdots, A_n)$ is called the \emph{weight distribution} of $\C$.
An $[n,\,\kappa,\,\omega]$ code $\C$ is called \emph{optimal} if its parameters
 $[n,\,\kappa,\,\omega]$ meet a bound on linear codes, and \emph{almost optimal}
 if $[n,\,\kappa,\,\omega+1]$ or $[n,\,\kappa+1,\,\omega]$ meets a bound on linear codes.

A generic construction of linear codes over finite fields has attracted a lot of attention in the past
eight  years (see, for example, \cite{Ding09}, \cite{Ding15}, \cite{DLN},\cite{DN07},\cite{DingDing1},
\cite{DingDing2}, and \cite{WDX15}). Many classes of one-weight, two-weight and three-weight codes
were obtained with this approach. The first goal of this paper is to establish a few relations among some classes
of linear codes obtained with this approach, so that the parameters of some classes of linear codes can be
derived from those of other classes with known parameters. In this way, linear codes with new parameters
will be derived. The second is to present a class of three-weight binary codes and consider their applications
in secret sharing.

\section{Group characters in $\gf(q)$}

An {\em additive character} of $\gf(q)$ is a nonzero function $\chi$
from $\gf(q)$ to the set of nonzero complex numbers such that
$\chi(x+y)=\chi(x) \chi(y)$ for any pair $(x, y) \in \gf(q)^2$.
For each $b\in \gf(q)$, the function
\begin{eqnarray}\label{dfn-add}
\chi_b(c)=\epsilon_p^{\tr(bc)} \ \ \mbox{ for all }
c\in\gf(q)
\end{eqnarray}
defines an additive character of $\gf(q)$, where and whereafter $\epsilon_p=e^{2\pi \sqrt{-1}/p}$ is
a primitive complex $p$th root of unity. When $b=0$,
$\chi_0(c)=1 \mbox{ for all } c\in\gf(q),
$
and is called the {\em trivial additive character} of
$\gf(q)$. The character $\chi_1$ in (\ref{dfn-add}) is called the
{\em canonical additive character} of $\gf(q)$.
It is known that every additive character of $\gf(q)$ can be
written as $\chi_b(x)=\chi_1(bx)$ \cite[Theorem 5.7]{LN}.

\section{A generic construction of linear codes}

Let $D=\{d_1, \,d_2, \,\ldots, \,d_n\} \subseteq \gf(q)$, where again $q=p^m$.
Let $\tr$ denote the trace function from $\gf(q)$ onto $\gf(p)$ throughout
this paper. We define a linear code of
length $n$ over $\gf(p)$ by
\begin{eqnarray}\label{eqn-maincode}
\C_{D}=\{(\tr(xd_1), \tr(xd_2), \ldots, \tr(xd_n)): x \in \gf(q)\},
\end{eqnarray}
and call $D$ the \emph{defining set} of this code $\C_{D}$. By definition, the
dimension of the code $\C_D$ is at most $m$.

The code $\C_D$ depends on the specific ordering of the elements in the defining set $D$. However,
up to column permutations, the codes obtained from different orderings are
\emph{equivalent} with respect to coordinate permutations. Hence, in this
paper, we do not specify the specific ordering of the elements in $D$ when we
consider the code $\C_D$.

This construction is generic in the sense that many classes of known codes
could be produced by properly selecting the defining set $D \subseteq \gf(q)$. This
construction technique was employed in \cite{Ding09},
\cite{Ding15},\cite{DLN}, \cite{DN07}, \cite{DingDing1},  \cite{DingDing2}, and \cite{WDX15}  for
obtaining linear codes with a few weights.

This construction is generic and can produce a lot of linear codes. The parameters
of the code $\C_D$ depend on the selection of the defining set $D$. Linear codes with
both poor and good error-correcting capability can be obtained with this approach.

Many classes of linear codes with a few weights and good parameters have been
already obtained with this approach. In this paper, we will present a few relations
among some subclasses of linear codes obtained with this approach. In this way,
we are able to derive the parameters of some other linear codes.

It is convenient to define for each $x \in \gf(q)$,
\begin{eqnarray}\label{eqn-mcodeword}
\bc_{x}=(\tr(xd_1), \,\tr(xd_2), \,\ldots, \,\tr(xd_n)).
\end{eqnarray}
The Hamming weight $\wt(\bc_x)$ of $\bc_x$ is $n-N_x(0)$, where
$$
N_x(0)=\left|\{1 \le i \le n: \tr(xd_i)=0\}\right|
$$
for each $x \in \gf(q)$.

It is easily seen that for any $D=\{d_1,\,d_2,\,\ldots, \,d_n\} \subseteq \gf(q)$
we have
\begin{eqnarray*}\label{eqn-hn3}
pN_x(0)
= \sum_{i=1}^n \sum_{y \in \gf(p)} e^{2\pi \sqrt{-1} y\tr(xd_i)/p}
= \sum_{i=1}^n \sum_{y \in \gf(p)} \chi_1(yxd_i)
= n + \sum_{y \in \gf(p)^*} \chi_1(yxD),
\end{eqnarray*}
where $\chi_1$ is the canonical additive character of $\gf(q)$, $aD$ denotes the set
$\{ad: d \in D\}$, and $\chi_1(S):=\sum_{x \in S} \chi_1(x)$ for any subset $S$ of $\gf(q)$.
Hence,
\begin{eqnarray}\label{eqn-weight}
\wt(\bc_x)=n-N_x(0)=\frac{(p-1)n-\sum_{y \in \gf(p)^*} \chi_1(yxD)}{p}.
\end{eqnarray}

\section{Shortening and expanding a linear code obtained from this construction}

Let $D \subset \gf(q)^*$ and $\hat{D}=E D$, where $E$ is a subset of $\gf(p)^*$ and
$$
ED=\{ed: e \in E \mbox{ and } d \in D\}.
$$
Let $n=|D|$, $\hat{n}=|\hat{D}|$, and $\ell=|E|$.

Our goal of this section is to establish a relation between the parameters of the two codes $\C_{D}$ and $\C_{\hat{D}}$
under the condition that $|ED|=|E||D|$. Specifically, we have the following general result.

\begin{theorem}\label{thm-mmm1}
Let symbols and notation be the same as above. Assume that $\hat{n}=n\ell$.  Then $\C_{D}$ is an $[n, k]$ linear code
with weight enumerator
$$
1+A_1z+A_2z^2+ \cdots + A_n z^n
$$
if and only if $\C_{\hat{D}}$ is an $[n\ell, k]$ linear code
with weight enumerator
$$
1+A_1z^{\ell}+A_2z^{2\ell}+ \cdots + A_n z^{n\ell}.
$$
\end{theorem}

\begin{proof}
Note that $\tr(zx)=z\tr(x)$ for all $z \in \gf(p)$ and $x \in \gf(q)$. We have that $\tr(zx)=0$ if and only if $\tr(x)=0$
for all $z \in \gf(p)^*$ and $x \in \gf(q)$.

Let $E=\{e_1, e_2, \cdots, e_\ell\}$. By assumption we have that $|ED|=|E||D|$.
Up to column permutations, every codeword $\hat{\bc}_x$ in $\C_{\hat{D}}$ can
be expressed as
$$
\hat{\bc}_x =(e_1 \bc_x, e_2 \bc_x, \cdots, e_\ell \bc_x),
$$
where $\bc_x$ is the corresponding codeword in $\C_D$. It then follows that $\wt(\hat{\bc}_x)=\ell \wt(\bc_x)$.
In addition, $\hat{\bc}_x$ is the zero codeword in $\C_{\hat{D}}$ if and only if $\bc_x$ is the zero
codeword in $\C_D$. The desired conclusions then follow.
\end{proof}

It should be noticed that the condition $|ED|=|E||D|$ in Theorem \ref{thm-mmm1} is necessary. Without this
condition, there may be no specific relation among the parameters of the two codes $\C_{D}$ and $\C_{\hat{D}}$.

Theorem \ref{thm-mmm1} can be employed to derive parameters of a shortened or expanded code of a linear code obtained
from this generic construction in some special cases. We now demonstrate this possibility in the rest of this section.

\begin{corollary}\label{cor-mmm11}
Let $\hat{D}=\gf(q)^*$, $E=\gf(p)^*$, and let $D$ be the coset representatives of the quotient group $\gf(q)^*/\gf(p)^*$.
Then $\hat{D}=ED$, and $\C_D$ is a one-weight code with parameters $[(q-1)/(p-1), m]$ and weight enumerator $1+(p^m-1)z^{p^{m-1}}$.
\end{corollary}

\begin{proof}
Note that $\hat{D}=\gf(q)^*$. For every $x \in \gf(q)^*$, the codeword $\hat{\bc}_x$ has Hamming weight $(p-1)p^{m-1}$.
Hence, $\C_{\hat{D}}$ is a $[p^m-1, m]$ code over $\gf(p)$ with the only nonzero weight $(p-1)p^{m-1}$. The desired
conclusions on $\C_D$ then follow from Theorem \ref{thm-mmm1}.
\end{proof}

The code $\C_D$ in Corollary \ref{cor-mmm11} is the well-known simplex code. The purpose of presenting this code is to demonstrate that it is a shortened version of the code $\C_{\hat{D}}$ obtained from the generic construction. In addition,
in the next corollary, we will show that this code can be extended into many ways so that many one-weight linear codes can
be obtained.

\begin{corollary}\label{cor-mmm12}
Let $D$ be the coset representatives of the quotient group $\gf(q)^*/\gf(p)^*$. Let $E$ be any subset of $\gf(p)^*$.
Define $\tilde{D}=ED$. Then $\C_{\tilde{D}}$ is a one-weight code with parameters $[\ell (q-1)/(p-1), m]$ and weight enumerator $1+(p^m-1)z^{\ell p^{m-1}}$,
where $\ell$ is the cardinality of $E$ and $1 \leq \ell \leq p-1$.
\end{corollary}

\begin{proof}
It was proved in Corollary \ref{cor-mmm11} that $\C_D$ is a one-weight code with parameters $[(q-1)/(p-1), m]$ and weight enumerator $1+(p^m-1)z^{p^{m-1}}$. By the definition of $D$, we have that $|ED|=|E||D|$.  The desired conclusions on the
code $\C_{\tilde{D}}$ then follow from Theorem \ref{thm-mmm1}.
\end{proof}

The codes $\C_{\tilde{D}}$ in Corollary \ref{cor-mmm12} are extensions and generalizations of the codes from skew sets
presented in \cite{Ding15}.  Since the total number of nonempty sets of $\gf(p)^*$ is $2^{p-1}-1$, the construction in
Corollary \ref{cor-mmm12} yields $2^{p-1}-1$ one-weight codes over $\gf(p)$.

We inform that the technique of shortening a linear code obtained from this generic construction was already employed in \cite{DingDing2}.
Theorem \ref{tab-HKMcodes} is a formal description and generalization of this technique.

\section{Combining two linear codes obtained from this generic construction}

\subsection{A method for combining two codes}

\begin{theorem}\label{thm-nnn1}
Let $D_1 \subset \gf(q)$ and $D_2 \subset \gf(q)$ with $D_1 \cap D_2=\emptyset$.
Define $D=D_1 \cup D_2$. Let $n_i=|D_i|$ for $i \in \{1,2\}$. Assume that $\C_D$ is
an $[n_1+n_2, k]$ one-weight code over $\gf(p)$ with nonzero-weight $w$, and $\C_{D_i}$
has also dimension $k$ for each $i$. Then $\C_{D_1}$ has weight enumerator
$$
1+A_1z+A_2z^2+ \cdots + A_{n_1} z^{n_1}
$$
if and only if
$\C_{D_2}$ has weight enumerator
$$
1+A_1z^{w-1}+A_2z^{w-2}+ \cdots + A_{n_1} z^{w-n_1}.
$$
\end{theorem}

\begin{proof}
By assumption, $D_1 \cap D_2=\emptyset$. It follows that up to coordinate permutations
every codeword $\bc_x$ in $\C_{D}$ can be expressed as
$$
\bc_x =\left(\bc_x^{(1)}, \bc_x^{(2)}\right),
$$
where $\bc_x^{(i)}$ is the codeword in $\C_{D_i}$. It then follows that $\wt(\bc_x)=\wt(\bc_x^{(1)})+\wt(\bc_x^{(2)})$.
By the assumptions of this theorem, all three codes have the same dimension. Hence, $\bc_x$ is the zero codeword in $\C_D$ if and only if
$\bc_x^{(1)}$ and $\bc_x^{(2)}$ are the zero codeword in $\C_{D_1}$ and $\C_{D_2}$ respectively. The desired
conclusions then follow.
\end{proof}

It should be pointed out that the condition  $D_1 \cap D_2=\emptyset$ in Theorem \ref{thm-nnn1} is necessary for the
correctness of the conclusion. This is implied by the proof of Theorem \ref{thm-nnn1} above.

Let $D \subset \gf(q)$. Starting from now on, let $\overline{D}$ denote the complement $\gf(q) \setminus D$ of $D$.
As a corollary of Theorem \ref{thm-nnn1}, we have the following.

\begin{corollary}\label{cor-LLL1}
Let $D \subset \gf(q)$. Assume that $\C_D$ is an $[n, m]$ linear code with
$$
\max_{\bc \in \C_D} \wt(\bc) < (p-1)p^{m-1}
$$
and weight enumerator
$$
1+A_1z+A_2z^2+ \cdots + A_{n} z^{n}.
$$
Then
$\C_{\overline{D}}$ has parameters $[q-n, m]$ and weight enumerator
$$
1+A_1z^{(p-1)p^{m-1}-1}+A_2z^{(p-1)p^{m-1}-2}+ \cdots + A_{n} z^{(p-1)p^{m-1}-n}.
$$
\end{corollary}

\begin{proof}
Note that $\C_{\gf(q)}$ is a $[q, m]$ linear code over $\gf(p)$, and has the only nonzero weight $(p-1)p^{m-1}$.
By definition, up to column permutations, every codeword $\bc$ in $\C_{\gf(q)}$ can be expressed as
$\bc=(\bc_x^{(1)}, \bc_x^{(2)})$, where $\bc_x^{(1)}$ and $\bc_x^{(2)}$ are the corresponding codewords in
$\C_{D}$ and $\C_{\overline{D}}$, respectively. Since $\C_D$ has dimension $m$, $\bc_x$ is the zero codeword in
$\C_{\gf(q)}$ if and only if $\bc_x^{(1)}$ is the zero codeword in $\C_{D}$. Note that $\C_{\gf(q)}$ is a one-weight
code with nonzero weight $(p-1)p^{m-1}$ and $\max_{\bc \in \C_D} \wt(\bc) < (p-1)p^{m-1}$. $\bc_x$ is a nonzero
codeword in $\C_{\gf(q)}$ if and only if $\bc_x^{(2)}$ is a nonzero codeword in $\C_{\overline{D}}$. We then deduce that
the dimension of $\C_{\overline{D}}$ is also $m$. The rest of the desired conclusions follows from Theorem \ref{thm-nnn1}.
\end{proof}

It is noticed that the condition
$
\max_{\bc \in \C_D} \wt(\bc) < (p-1)p^{m-1}
$
in Corollary \ref{cor-LLL1} is necessary. Without this condition, the dimension of $\C_{\overline{D}}$ may be less
than $m$.

Corollary \ref{cor-LLL1} could be employed to determine the weight enumerator of many classes of linear codes
$\C_{\overline{D}}$ from those of the code $\C_D$. In the next subsections, we will demonstrate this with
a few examples.

\subsection{A class of one-weight and two-weight codes}

Let $f(x)$ be a function from $\gf(q)$ to $\gf(q)$. We define
$$
D(f):=\{f(x): x \in \gf(q)\} \setminus \{0\}.
$$

A polynomial $f$ over $\gf(q)$ of the form
\begin{eqnarray}\label{eqn-quadraticforms}
f(x)=\sum_{i \in I} \sum_{j \in J} a_{i,j} x^{p^i+p^j}
\end{eqnarray}
is called a {\em quadratic form} over $\gf(q)$, where $a_{i,j} \in \gf(q)$,  and
$I$ and $J$ are subsets of $\{0,1,2, \ldots, m-1\}$.

Note that $\gf(q)$ is a vector space of dimension $m$ over $\gf(p)$. The rank of the quadratic form $f$ over $\gf(q)$
is defined to be the codimension of the $\gf(p)$-vector space
$$
V_f=\{x \in \gf(q): f(x+z)-f(x)-f(z)=0 \mbox{ for all } z \in \gf(q)\}.
$$
That is $|V_f|=p^{m-r}$, where $r$ denotes the rank of $f$.

It is still very difficult to determine the length $n_f$ of the code $\C_{D(f)}$ for general quadratic forms $f$,
let alone the weight distribution of the code $\C_{D(f)}$. However, under certain conditions the
weight distribution of $\C_{D(f)}$ can be worked out \cite{Ding15}. Below we derive a general
result on the code $\C_{\overline{D(f)}}$ from known results on the code $\C_{D(f)}$.

\begin{corollary}\label{cor-qfcodes}
Let $f$ be a quadratic form of rank $r$ over $\gf(q)$ such that
\begin{itemize}
\item $f(0)=0$ and $f(x) \neq 0$ for all $x \in \gf(q)^*$; and
\item $f$ is $e$-to-$1$ on $\gf(q)^*$ (i.e. $f(x)=u$ has either $e$ solutions $x \in \gf(q)^*$ or no solution for each $u \in \gf(q)^*$),
          where $e$ is a positive integer.
\end{itemize}

If $r$ is odd and $e>1$, then $\C_{\overline{D(f)}}$ is a one-weight code over $\gf(p)$ with parameters
$$
\left[\frac{(e-1)q+1}{e}, \  m, \  \frac{(e-1)(p-1)q}{ep}\right].
$$

If $r \geq 2$ is even and $e>1$, then $\C_{\overline{D(f)}}$ is a two-weight code over $\gf(p)$ with parameters
$$
\left[\frac{(e-1)q+1}{e},\  m, \  \frac{(p-1)((e-1)q-p^{m-r/2})}{ep}\right]
$$
and weight enumerator
\begin{eqnarray*}
1+ \frac{q-1}{2} z^{\frac{(p-1)((e-1)q-p^{m-r/2})}{ep}} + \frac{q-1}{2} z^{\frac{(p-1)((e-1)q+p^{m-r/2})}{ep}}.
\end{eqnarray*}
\end{corollary}

\begin{proof}
The parameters of the code $\C_{D(f)}$ were determined in \cite{Ding15}. In addition, it is easily verified that
the conditions of Corollary \ref{cor-LLL1} were satisfied by $\C_{D(f)}$. The desired conclusions on the code
$\C_{\overline{D(f)}}$ then follow from Corollary \ref{cor-LLL1} and Theorem 3 in \cite{Ding15}.
\end{proof}

\begin{example}
$f(x)=x^{p^\ell+1}$ is a quadratic form over $\gf(q)$ satisfying the conditions of Corollary \ref{cor-qfcodes},
if $e=\gcd(q-1, p^\ell +1)>1$.
\end{example}

\begin{example}
$f(x)=x^{10} -ux^{6} -u^2x^2$ is a quadratic form over $\gf(3^m)$ satisfying the conditions of Corollary \ref{cor-qfcodes},
where $u \in \gf(3^m)$, $m$ is odd, and $e=2$.
\end{example}

\subsection{The parameters of some binary codes}

Let $f$ be a Boolean function from $\gf(2^m)$ to $\gf(2)$. The \emph{support} of $f$ is defined to be
$$
D_f=\{x \in\gf(2^m) : f(x)=1\} \subseteq \gf(2^m).
$$
Recall that $n_f=|D_f|$.

The {\em Walsh transform} of $f$ is defined by
\begin{eqnarray}\label{eqn-WalshTransform2}
\hat{f}(w)=\sum_{x \in \gf(2^m)} (-1)^{f(x)+\tr(wx)}
\end{eqnarray}
where $w \in \gf(2^m)$. The {\em Walsh spectrum} of $f$ is the following multiset
$$
\left\{\left\{ \hat{f}(w): w \in \gf(2^m) \right\}\right\}.
$$

It is in general a very hard to determine the weight distribution of the binary code $\C_{D_f}$ with length $n_f$
and dimension at most $m$. However, in a number of special cases, this can be done. Below we describe the
weight distribution of the code $\C_{\overline{D_f}}$ by making use of some results on the binary code $\C_{D_f}$
obtained in \cite{Ding15}.

The main result of this section is described in the following corollary.

\begin{corollary}\label{cor-BooleanCodes}
Let symbols and notation be the same as above. If $2n_f \neq \hat{f}(w)$ for all $w \in \gf(2^m)^*$
and
\begin{eqnarray}\label{eqn-condb}
\max_{w \in \gf(2^m)} \hat{f}(w) < 2(2^m-n_f),
\end{eqnarray}
then $\C_{\overline{D_f}}$ is a binary linear code with
length $2^m-n_f$ and dimension $m$, and its weight distribution is given by the following multiset:
\begin{eqnarray}\label{eqn-WTBcodes}
\left\{\left\{ \frac{2(2^m-n_f)-\hat{f}(w)}{4}: w \in \gf(2^m)^*\right\}\right\} \cup \left\{\left\{ 0 \right\}\right\}.
\end{eqnarray}
\end{corollary}

\begin{proof}
According to Theorem 9 in \cite{Ding15}, the code $\C_{D_f}$ has dimension $m$ and the maximum nonzero
Hamming weight of codewords in $\C_{D_f}$ is
$$
\max_{w \in \gf(2^m)} \frac{2n_f+\hat{f}(w)}{4} < 2^{m-1}
$$
due to the condition of (\ref{eqn-condb}).  The desired conclusions then follow from Corollary \ref{cor-LLL1} and
Theorem 9 in \cite{Ding15}.
\end{proof}

For all the two-weight and three-weight codes $\C_{D_f}$ from bent, semibent and almost bent functions described
in \cite{Ding15}, the corresponding binary codes  $\C_{\overline{D_f}}$ are also two-weight and three-weight codes
with the weight distribution given in Corollary \ref{cor-BooleanCodes}. We omit the details of the weight distributions
here.

\subsection{A class of ternary codes}

In this subsection, we determine the parameters of a class of ternary codes.

\begin{center}
\begin{table}[ht]
\caption{The weight distribution of the codes of Theorem \ref{thm-HKMcodes}}\label{tab-HKMcodes}
\begin{tabular}{ll} \hline
Weight $w$ &  Multiplicity $A_w$  \\ \hline
$0$          &  $1$ \\ \hline
$5 \times 3^{3h-2}+3^{2h-2}$  & $3^{2h}+3^{h}$ \\ \hline
$5 \times 3^{3h-2}$  & $3^{3h}-2 \times 3^{2h}-1$ \\ \hline
$5 \times 3^{3h-2}-3^{2h-2}$  & $3^{2h}-3^{h}$ \\ \hline
\end{tabular}
\end{table}
\end{center}

Our main result of this section is the following.

\begin{theorem}\label{thm-HKMcodes}
Let $h$ be an odd positive integer and let $m=3h \geq 3$. Let $\ell=3^{2h}-3^h+1$. Define
\begin{eqnarray}\label{eqn-HKMds}
D=\left\{ \alpha^t: \tr_{3^m/3}(\alpha^t + \alpha^{t\ell})=0, \ 0 \leq t \leq (3^m-3)/2  \right\},
\end{eqnarray}
where $\alpha$ is a generator of $\gf(3^m)^*$.

Then the ternary code
$\C_{\overline{D}}$ has parameters
$$
\left[ \frac{5 \times 3^{3h-1}+1}{2}, \ 3h, \ 5 \times 3^{3h-2}-3^{2h-2}\right]
$$
and the weight distribution of Table \ref{tab-HKMcodes}.
\end{theorem}

\begin{proof}
According to Theorem 15 in \cite{Ding15}, the ternary code $\C_{D}$ has dimension $m=3h$ and the maximum nonzero
Hamming weight of codewords in $\C_{D_f}$ is
$$
3^{3h-2}+3^{2h-2} < 2 \times 3^{3h-1}.
$$
The desired conclusions then follow from Corollary \ref{cor-LLL1} and
Theorem 15 in \cite{Ding15}.
\end{proof}

We remark that the code $\C_{\overline{D}}$ of Theorem \ref{thm-HKMcodes} has more than three nonzero weights if $h$ is
even.

\subsection{Other two-weight and three-weight codes}

Corollary \ref{cor-LLL1}  can be employed to obtain many linear codes with a few weights. For instances, it can be
applied to all the codes in \cite{DingDing1}, \cite{DingDing2}, and \cite{WDX15} to obtain two-weight and three-weight
codes.

\section{A class of binary linear codes and their parameters}

The objective of this section is to present a class of three-weight binary codes and consider their application
in secret sharing. In this section, let $p=2$ and we consider only binary codes $\C_D$ within the generic
construction of this paper.

\subsection{The description of the codes}\label{sec-3wtcodes}

In this subsection, we only describe the binary codes and introduce their parameters, but will present the proofs
of their parameters in the next subsection.

In this subsection, the defining set $D$ of the code $\C_D$ of (\ref{eqn-maincode})  is given by
\begin{eqnarray}\label{eqn-defsetD}
D=\{x \in \gf(2^m)^*: \tr(x^{3}+x)=0\}.
\end{eqnarray}

Since $0 \not\in D$, the minimum distance $d^{\perp}$ of the dual code $\C_D^{\perp}$ of $\C_D$ cannot be $1$.
Note that the elements in $D$ are pairwise distinct, the minimum distance $d^{\perp}$ of the dual code $\C_D^{\perp}$
cannot be $2$. Hence, we have the following lemma.

\begin{lemma}\label{lem-duald}
The minimum distance $d^{\perp}$ of the dual code $\C_D^{\perp}$ of $C_D$ is at least $3$ if $n=|D| \geq 3$.
\end{lemma}

\begin{theorem}\label{thm-twothree1}
Let $m\geq 5$ be odd, and let $D$ be defined in (\ref{eqn-defsetD}). Then the set $\C_D$ of (\ref{eqn-maincode}) is a
$[2^{m-1}-1+(-1)^{\frac{m^2-1}{8}}2^{\frac{m-1}{2}}, m]$ binary code
with the weight distribution in Table \ref{tab-twothree1}. The dual code $\C_D^{\perp}$ has parameters
$$
[2^{m-1}-1+(-1)^{\frac{m^2-1}{8}}2^{\frac{m-1}{2}}, 2^{m-1}-1+(-1)^{\frac{m^2-1}{8}}2^{\frac{m-1}{2}}-m,
d^{\perp}],
$$
where $d^{\perp} \geq 3$.
\end{theorem}

\begin{table}[ht]
\caption{The weight distribution of the codes $\C_D$ of Theorem \ref{thm-twothree1}}\label{tab-twothree1}
\begin{tabular}{ll} \hline
Weight $w$ &  Multiplicity $A_w$  \\ \hline
$0$          &  $1$ \\ \hline
$2^{m-2}+[(-1)^{\frac{m^2-1}{8}}-1]2^{\frac{m-3}{2}}$ & $2^{m-2}+2^{\frac{m-3}{2}}-\frac{1}{2}[1+(-1)^{\frac{m^2-1}{8}}]$ \\ \hline
$2^{m-2}+(-1)^{\frac{m^2-1}{8}}2^{\frac{m-3}{2}}$    & $2^{m-1}$ \\ \hline
$2^{m-2}+[(-1)^{\frac{m^2-1}{8}}+1]2^{\frac{m-3}{2}}$ & $2^{m-2}-2^{\frac{m-3}{2}}+\frac{1}{2}[-1+(-1)^{\frac{m^2-1}{8}}]$ \\ \hline
\end{tabular}
\end{table}

\begin{example}
Let $m=5$. Then the code $\C_D$ has parameters $[11, 5, 4]$ and weight enumerator
$1+10z^{4}+16z^{6}+5z^{8}$. This code is optimal. Its dual $\C_D^\perp$ has parameters
$[11, 6, 3]$ and is almost optimal.
\end{example}

\begin{example}
Let $m=7$. Then the code $\C_D$ has parameters $[71, 7, 32]$ and weight enumerator
$1+35z^{32}+64z^{36}+28z^{40}$. Its dual $\C_D^\perp$ has parameters
$[71, 64, 3]$ and is optimal.
\end{example}

\begin{theorem}\label{thm-twothree2}
Let $m\geq 4$ be even, and let $D$ be defined in (\ref{eqn-defsetD}). Then the set $\C_D$ of (\ref{eqn-maincode}) is a
$[2^{m-1}-1, m]$ binary code with the weight distribution in Table \ref{tab-twothree2}
when $m \equiv 2 \inmod{4}$, and  a
$[2^{m-1}-1-2^{\frac{m}{2}}(-1)^{\frac{m}{4}}, m]$ binary code
with the weight distribution in Table \ref{tab-twothree3} when $m\equiv 0 \inmod{4}$.

The dual code $\C_D^{\perp}$
has parameters  $[2^{m-1}-1, 2^{m-1}-1-m, d^{\perp} \geq 3]$
when $m \equiv 2 \inmod{4}$, and  parameters
$$
[2^{m-1}-1-2^{\frac{m}{2}}(-1)^{\frac{m}{4}}, 2^{m-1}-1-2^{\frac{m}{2}}(-1)^{\frac{m}{4}}-m, d^{\perp}\geq 3]
$$
when $m\equiv 0 \inmod{4}$.
\end{theorem}

\begin{table}[ht]
\caption{The weight distribution of the codes $\C_D$ of Theorem \ref{thm-twothree2} when $m \equiv 2 \inmod{4}$}\label{tab-twothree2}
\begin{tabular}{ll} \hline
Weight $w$ &  Multiplicity $A_w$  \\ \hline
$0$          &  $1$ \\ \hline
$2^{m-2}$  & $3\cdot2^{m-2}-1 $ \\ \hline
$2^{m-2}+ 2^{\frac{m-2}{2}}$   &$2^{m-3} - 2^{\frac{m-4}{2}} $ \\ \hline
$2^{m-2}- 2^{\frac{m-2}{2}}$   &$2^{m-3} + 2^{\frac{m-4}{2}} $ \\ \hline
\end{tabular}
\end{table}

\begin{table}[ht]
\caption{The weight distribution of the codes $\C_D$ of Theorem \ref{thm-twothree2} \ when $m \equiv 0 \inmod{4}$}\label{tab-twothree3}
\begin{tabular}{ll} \hline
Weight $w$ &  Multiplicity $A_w$  \\ \hline
$0$          &  $1$ \\ \hline
$2^{m-2}-(-1)^{\frac{m}{4}}2^{\frac{m-2}{2}}$  & $3\cdot2^{m-2}$ \\ \hline
$2^{m-2}-[(-1)^{\frac{m}{4}}+1]2^{\frac{m-2}{2}}$ & $2^{m-3} + 2^{\frac{m-4}{2}}+ \frac{1}{2}[(-1)^{\frac{m}{4}}-1] $ \\ \hline
$2^{m-2}-[(-1)^{\frac{m}{4}}-1]2^{\frac{m-2}{2}}$ & $2^{m-3} - 2^{\frac{m-4}{2}}- \frac{1}{2}[(-1)^{\frac{m}{4}}+1] $ \\ \hline
\end{tabular}
\end{table}

\begin{example}
Let $m=6$. Then the code $\C_D$ has parameters $[31, 6, 12]$ and weight enumerator
$1+10z^{12}+47z^{16}+6z^{20}$. Its dual $\C_D^\perp$ has parameters
$[31, 25, 3]$ and is almost optimal.
\end{example}

\begin{example}
Let $m=10$. Then the code $\C_D$ has parameters $[511, 10, 240]$ and weight enumerator
$1+136z^{240}+767z^{256}+120z^{272}$. Its dual $\C_D^\perp$ has parameters
$[511, 501, 3]$.
\end{example}

\begin{example}
Let $m=4$. Then the code $\C_D$ has parameters $[11, 4, 4]$ and weight enumerator
$1+2z^{4}+12z^{6}+z^{8}$. This code is almost optimal. Its dual $\C_D^\perp$ has parameters
$[11, 7, 3]$ and is optimal.
\end{example}

\begin{example}
Let $m=8$. Then the code $\C_D$ has parameters $[111, 8, 48]$ and weight enumerator
$1+36z^{48}+192z^{56}+27z^{64}$. Its dual $\C_D^\perp$ has parameters
$[111, 103, 3]$ and is almost optimal.
\end{example}

\subsection{The proofs of Theorems \ref{thm-twothree1} and \ref{thm-twothree2}}

For any $a$ and $b$ in $\gf(q)$, we define the following exponential sum
\begin{eqnarray}\label{eqn-esum}
S(a,b)=\sum_{x \in \gf(q)} \chi_1\left(a x^{3}+bx\right).
\end{eqnarray}
To prove the weight distributions of the codes in Theorems \ref{thm-twothree1} and
\ref{thm-twothree2}, we need the values of the sum $S(a,b)$.

We now define a constant as follows. Let
$$
n_0=\left|\left\{x \in \gf(q): \tr\left(x^{3}+x\right)=0\right\}\right|.
$$
By definition, the length $n$ of the code $\C_D$ of (\ref{eqn-maincode}) is equal to $n_0-1$.
We have
\begin{eqnarray}\label{eqn-nzero}
n_0
&=& \frac{1}{2} \sum_{y \in \gf(2)} \sum_{x \in \gf(q)} (-1)^{y\tr(x^{3}+x)} \nonumber \\
&=& \frac{1}{2} \sum_{y \in \gf(2)} \sum_{x \in \gf(q)} \chi_1(y x^{3}+y x) \nonumber \\
&=& 2^{m-1} +  \frac{1}{2} \sum_{x \in \gf(q)} \chi_1(x^{3}+x).
\end{eqnarray}

To prove Theorems \ref{thm-twothree1} and \ref{thm-twothree2}, we also define the following
parameter
$$
N_{b}=\left|\left\{x \in \gf(q): \tr\left(x^{3}+x\right)=0 \mbox{ and } \tr(bx)=0\right\}\right|,
$$
where $b \in \gf(q)^*$.
By definition and the basic facts of additive characters, for any $b \in \gf(q)^*$ we have
\begin{eqnarray}\label{eqn-Nb0}
N_{b}
&=& \frac{1}{4} \sum_{x \in \gf(q)} \left( \sum_{y \in \gf(2)} (-1)^{y\tr(x^{3}+x)} \right)
                                                   \left( \sum_{z \in \gf(2)}  (-1)^{z\tr(bx)} \right) \nonumber \\
&=&  \frac{1}{4}   \sum_{x \in \gf(q)}   (-1)^{\tr(bx)}  +
         \frac{1}{4}    \sum_{x \in \gf(q)}   (-1)^{\tr(x^{3}+x)} + \nonumber \\
& &  \frac{1}{4}     \sum_{x \in \gf(q)}   (-1)^{\tr(x^{3}+(b+1)x)} + 2^{m-2} \nonumber \\
& =&         \frac{1}{4}    \left(\sum_{x \in \gf(q)}  \left( \chi_1(x^{3}+x) +
          \chi_1(x^{3}+(b+1)x)  \right) + 2^{m}  \right).
\end{eqnarray}

For any $b \in \gf(q)^*$, the Hamming weight $\wt(\bc_b)$ of the following codeword
\begin{eqnarray}\label{eqn-codewordb}
\bc_b=(\tr(bd_1), \tr(bd_2), \ldots, \tr(bd_n))
\end{eqnarray}
of the code $\C_D$ of (\ref{eqn-maincode}) is equal to $n_0-N_{b}$.

Let $m\geq 4$ be odd. It is well known that $\tr(x^{3})$ is a semibent function from $\gf(q)$ to $\gf(2)$. Thus, we have

\begin{eqnarray}\label{eqn-semibent}
\sum_{x \in \gf(q)} \chi_1(x^{3}+(b+1)x))\in \{0,2^{\frac{m+1}{2}},-2^{\frac{m+1}{2}}\}
\end{eqnarray}
for each $b \in \gf(q)^*$.

The following lemma is proved in Theorem 2 of \cite{Carlitz}.

\begin{lemma}\label{lem-32A1}
When $m$ is odd, we have
$$
S(1,1)=\sum_{x \in \gf(q)} \chi_1 \left(x^{3}+x\right)=(-1)^{\frac{m^{2}-1}{8}}2^{\frac{m+1}{2}}.
$$
\end{lemma}

We are now ready to prove Theorem \ref{thm-twothree1}. Let $m\geq 4$ be odd.

It follows from (\ref{eqn-nzero}) and Lemma \ref{lem-32A1} that the length $n$ of the code $\C_D$ in Theorem \ref{thm-twothree1} is equal to
$2^{m-1}+(-1)^{\frac{m^{2}-1}{8}}2^{\frac{m-1}{2}}-1$, as $n_0=2^{m-1}+(-1)^{\frac{m^{2}-1}{8}}2^{\frac{m-1}{2}}$.

It follows from (\ref{eqn-Nb0}), (\ref{eqn-semibent}) and Lemma \ref{lem-32A1} that
$$
N_{b} \in \left\{ 2^{m-2}+ (-1)^{\frac{m^{2}-1}{8}}2^{\frac{m-3}{2}}, \, 2^{m-2}+ [(-1)^{\frac{m^{2}-1}{8}}\pm 1]2^{\frac{m-3}{2}} \right\}
$$
for any $b \in \gf(q)^*$. Hence, the weight $\wt(\bc_b)$ of the codeword $\bc_b$ in (\ref{eqn-codewordb}) satisfies
$$
\wt(\bc_b)=n_0-N_{b}
\in \left\{ 2^{m-2}+ (-1)^{\frac{m^{2}-1}{8}}2^{\frac{m-3}{2}}, \, 2^{m-2}+ [(-1)^{\frac{m^{2}-1}{8}}\mp 1]2^{\frac{m-3}{2}} \right\}.
$$

Define
\begin{eqnarray*}
&& w_1=2^{m-2}+ (-1)^{\frac{m^{2}-1}{8}}2^{\frac{m-3}{2}}, \\
&& w_2=2^{m-2}+ [(-1)^{\frac{m^{2}-1}{8}}- 1]2^{\frac{m-3}{2}}, \\
&& w_3=2^{m-2}+ [(-1)^{\frac{m^{2}-1}{8}} + 1]2^{\frac{m-3}{2}}.
\end{eqnarray*}

We now determine the number $A_{w_i}$ of codewords with weight $w_i$ in $\C_{D}$.
By Lemma \ref{lem-duald}, the minimum weight of the dual code $\C_{D}^\perp$ is at least $3$.
The first three Pless Power Moments \cite[p.260]{HP} lead to the following system of equations:
\begin{eqnarray}\label{eqn-wtdsemibentfcode5}
\left\{
\begin{array}{lll}
A_{w_1}+A_{w_2}+A_{w_3} &=& 2^m-1, \\
w_1A_{w_1}+w_2A_{w_2}+w_3A_{w_3} &=& n 2^{m-1}, \\
w_1^2A_{w_1}+w_2^2A_{w_2}+w_3^2A_{w_3} &=& n(n+1) 2^{m-2},
\end{array}
\right.
\end{eqnarray}
where $n=2^{m-1}+(-1)^{\frac{m^{2}-1}{8}}2^{\frac{m-1}{2}}-1$.
Solving the system of equations in (\ref{eqn-wtdsemibentfcode5}) yields
the weight distribution of Table \ref{tab-twothree1}.
The dimension of the code $\C_D$ is $m$, as $\wt(\bc_b)>0$ for each $b \in \gf(q)^*$.
The conclusions on the dual code $\C_D^{\perp}$ then follow on the length and the
dimension of $\C_D$ and Lemma \ref{lem-duald}.
This completes the proof of Theorem \ref{thm-twothree1}.

Below we prove Theorem \ref{thm-twothree2}

Let $m\geq 4$ be even. To prove Theorem \ref{thm-twothree2}, we need the next two lemmas proved by
Coulter \cite{Coult}.

\begin{lemma}\label{lem-32A1c}
Let $m\geq 4$ be even and $a\in \gf(q)^*$. Then
\begin{eqnarray*}
S(a, 0)=\left\{
\begin{array}{l}
(-1)^{\frac{m}{2}} 2^{\frac{m}{2}}       ~~~~~~\mbox{ if $a \ne g^{3t}$ for any $t$,} \\
-(-1)^{\frac{m}{2}} 2^{\frac{m}{2}+1}      ~\mbox{ if $a = g^{3t}$ for some $t$,}
\end{array}
\right.
\end{eqnarray*}
where $g$ is a generator of $\gf(q)^*$.
\end{lemma}

\begin{lemma}\label{lem-32A2c}
Let  $m\geq 4$ be even, $b \in \gf(q)^*$, $f(x)=a^{2} x^{4} +ax \in \gf(q)[x]$, and let $g$ be a generator of
$\gf(q)^*$. There are the following two cases.

(i) If $a \ne g^{3t}$ for any $t$, then $f$ is a permutation polynomial of $\gf(q)$.
Let $x_0$ be the unique element satisfying $f(x_0)=b^{2}$. Then
$$
S(a,b)=(-1)^{\frac{m}{2}} 2^{\frac{m}{2}} \chi_1(a x_0^{3})=(-1)^{\frac{m}{2}} 2^{\frac{m}{2}} (-1)^{\tr(a x_0^{3})}.
$$

(ii) If $a = g^{3t}$ for some $t$, then $S(a, b)=0$ unless the equation $f(x)=b^{2}$
is solvable. If this equation is solvable, with solution $x_0$ say, then
\begin{eqnarray*}
S(a,b)=\left\{ \begin{array}{ll}
-(-1)^{\frac{m}{2}} 2^{\frac{m}{2}+1} (-1)^{\tr(a x_0^{3})}  & \mbox{ if $\tr(a)=0$,} \\
(-1)^{\frac{m}{2}} 2^{\frac{m}{2}} (-1)^{\tr(a x_0^{3})}  & \mbox{ if $\tr(a)\ne 0$,} \\
\end{array}
\right.
\end{eqnarray*}
where $\tr$ is the trace function from $\gf(q)$ onto $\gf(2)$.
\end{lemma}

According to \cite[p.29]{Dickson}, the following lemma can be easily proved.
\begin{lemma}\label{lem-32A2cc}
Let  $m\geq 4$ be even and $f(x)=a^{2} x^{4} +ax \in \gf(q)[x]$. If $a =1= g^{3t}$ for some $t$, then the equation $f(x)=1$
is solvable if and only if $m\equiv 0 \inmod{4}$, where $g$ is a generator of $\gf(q)^*$.
\end{lemma}

The next lemma will be employed later.

\begin{lemma}\label{lem-s11}
Let  $m\geq 4$ be even. Then

\begin{eqnarray*}
S(1,1)=\left\{ \begin{array}{ll}
0  & \mbox{ if $m \equiv 2 \inmod{4}$,} \\
-(-1)^{\frac{m}{4}} 2^{\frac{m}{2}+1} & \mbox{ if $m \equiv 0 \inmod{4}$.} \\
\end{array}
\right.
\end{eqnarray*}
\end{lemma}

\begin{proof}
Let $m\geq 4$ be even. It is well known that $\gcd(3,2^{m}-1)=3$. Hence, there exists $t=\frac{2^{m}-1}{3}$ such that $g^{3t}=1$. Note that $\tr(1)=0$, as $m$ is even. It then follows from Lemmas \ref{lem-32A2c} and \ref{lem-32A2cc} that

\begin{eqnarray*}
S(1,1)
&=& \left\{ \begin{array}{ll}
0                                        & \mbox{ if $m \equiv 2 \inmod{4}$} \\
-2^{\frac{m}{2}+1} (-1)^{\tr(x_{0}^{3})} & \mbox{ if $m \equiv 0 \inmod{4}$}
\end{array}
\right. \\
&=& \left\{ \begin{array}{ll}
0                                        & \mbox{ if $m \equiv 2 \inmod{4}$} \\
-2^{\frac{m}{2}+1} (-1)^{\frac{m}{4} \inmod{2}} & \mbox{ if $m \equiv 0 \inmod{4}$}
\end{array}
\right.\\
&=& \left\{ \begin{array}{ll}
0                                        & \mbox{ if $m \equiv 2 \inmod{4}$,} \\
-2^{\frac{m}{2}+1} (-1)^{\frac{m}{4}} & \mbox{ if $m \equiv 0 \inmod{4}$,}
\end{array}
\right.
\end{eqnarray*}
where $x_0$ is a solution of the equation $x^{4}+x=1$ when $m \equiv 0 \inmod{4}$.
This completes the proof.
\end{proof}

We are now ready to prove Theorem \ref{thm-twothree2}. Recall that $m\geq 4$ is even.
It follows from (\ref{eqn-nzero}) and Lemma
\ref{lem-s11} that the length $n$ of the code $\C_D$ in Theorem \ref{thm-twothree2} is given by
\begin{eqnarray}\label{eqn-nnnnn}
n = \left\{ \begin{array}{ll}
  2^{m-1}-1                                        & \mbox{ if $m \equiv 2 \inmod{4},$} \\
  2^{m-1}-2^{\frac{m}{2}} (-1)^{\frac{m}{4}}-1    & \mbox{ if $m \equiv 0 \inmod{4}.$}
\end{array}
\right.
\end{eqnarray}

Since $\gcd(3, 2^{m}-1)=3$, there exists $t=\frac{2^{m}-1}{3}$ such that $g^{3t}=1$. Note that $\tr(1)=0$, as $m$ is even.
It follows from Lemmas \ref{lem-32A1c} and \ref{lem-32A2c} that
\begin{eqnarray}\label{eqn-s1b}
S(1,b+1) \in \{0, \pm (-1)^{\frac{m}{2}} 2^{\frac{m}{2}+1} \}
\end{eqnarray}
for any $b \in \gf(q)^*$.

It then follows from (\ref{eqn-Nb0}), (\ref{eqn-s1b}) and Lemma \ref{lem-s11} that
\begin{eqnarray*}
N_{b} \in \{u_{1},~\pm u_{2}+u_{1}\}
\end{eqnarray*}
when $m \equiv 2 \inmod{4}$, and
\begin{eqnarray*}
N_{b} \in \left\{u_{1}-(-1)^{\frac{m}{4}}u_{2},(-(-1)^{\frac{m}{4}}\pm 1)u_{2}+u_{1}\right\}
\end{eqnarray*}
when  $m \equiv 0 \inmod{4}$,
for any $b \in \gf(q)^*$, where $u_{1}=2^{m-2}$ and $u_{2}=2^{\frac{m}{2}-1}$. Hence, the weight $\wt(\bc_b)$ of the codeword of (\ref{eqn-codewordb}) satisfies
\begin{eqnarray*}
\wt(\bc_b)=n_0-N_{b} \in
 \left\{\begin{array}{l}
\{u_{1},u_{1}\pm u_{2}\}                  \mbox{ if $m \equiv 2 \inmod{4} $} \\
\{u_{1}-(-1)^{\frac{m}{4}}u_{2},u_{1}-((-1)^{\frac{m}{4}}\pm 1)u_{2}\}                 \mbox{ if $m \equiv 0 \inmod{4} $} \\
\end{array}
\right.
\end{eqnarray*}
and the code $\C_D$ has all the three weights in the set above.

Define $u_{1}=2^{m-2}$, $u_{2}=2^{\frac{m}{2}-1}$, $u=u_{1}-(-1)^{\frac{m}{4}}u_{2}$ and
\begin{eqnarray*}
\left\{\begin{array}{ll}
w_1=u_{1},\ w_2=u_{1}+u_{2},\ w_3=u_{1}-u_{2}                & \mbox{if $m \equiv 2 \inmod{4} $} \\
w_1=u,\ w_2=u-u_{2},\ w_3=u+u_{2}   & \mbox{if $m \equiv 0 \inmod{4} $}.
\end{array}
\right.
\end{eqnarray*}

We now determine the number $A_{w_i}$ of codewords with weight $w_i$ in $\C_{D}$.
By Lemma \ref{lem-duald}, the minimum weight of the dual code $\C_{D}^\perp$ is at least $3$.
The first three Pless Power Moments \cite[p.260]{HP} lead to the following system of equations:
\begin{eqnarray}\label{eqn-wtdsemibentfcode6c}
\left\{
\begin{array}{lll}
A_{w_1}+A_{w_2}+A_{w_3} &=& 2^m-1, \\
w_1A_{w_1}+w_2A_{w_2}+w_3A_{w_3} &=& n 2^{m-1}, \\
w_1^2A_{w_1}+w_2^2A_{w_2}+w_3^2A_{w_3} &=& n(n+1) 2^{m-2},
\end{array}
\right.
\end{eqnarray}
where $n$ is given in (\ref{eqn-nnnnn}).
Solving the system of equations in (\ref{eqn-wtdsemibentfcode6c}) proves the weight distribution
of the code $\C_D$ in Table \ref{tab-twothree2} and Table \ref{tab-twothree3}.
The dimension of the code $\C_D$ is $m$, as $\wt(\bc_b)>0$ for each $b \in \gf(q)^*$.
The conclusions on the dual code $\C_D^{\perp}$ then follow on the length and the
dimension of $\C_D$ and Lemma \ref{lem-duald}.
This completes the proof of Theorem \ref{thm-twothree2}.

\subsection{Applications of the binary codes in secret sharing}

Any linear code over $\gf(p)$ can be employed to construct secret sharing schemes \cite{ADHK,CDY05,YD06}. In order to
obtain secret sharing schemes with interesting access structures, one would like to have linear codes $\C$ such that
$w_{\min}/w_{\max} > \frac{p-1}{p}$ \cite{YD06}, where $w_{\min}$ and $w_{\max}$ denote the minimum and maximum
nonzero weight of the linear code.

When $m \equiv 2 \inmod{4}$ and $m \geq 6$, the code $\C_D$ of Section \ref{sec-3wtcodes} satisfies that
\begin{eqnarray*}
\frac{w_{\min}}{w_{\max}} = \frac{2^{m-2}-2^{(m-2)/2}}{2^{m-2}+2^{(m-2)/2}} > \frac{1}{2}.
\end{eqnarray*}

When $m \equiv 4 \inmod{8}$ and $m >4$, the code $\C_D$ of Section \ref{sec-3wtcodes} satisfies that
\begin{eqnarray*}
\frac{w_{\min}}{w_{\max}} = \frac{2^{m-2}}{2^{m-2}+2^{m/2}} > \frac{1}{2}.
\end{eqnarray*}

When $m \equiv 0 \inmod{8}$ and $m \geq 8$, the code $\C_D$ of Section \ref{sec-3wtcodes} satisfies that
\begin{eqnarray*}
\frac{w_{\min}}{w_{\max}} = \frac{2^{m-2}-2^{m/2}}{2^{m-2}} > \frac{1}{2}.
\end{eqnarray*}

When $m \equiv \pm 1 \inmod{8}$ and $m \geq 7$, the code $\C_D$ of Section \ref{sec-3wtcodes} satisfies that
\begin{eqnarray*}
\frac{w_{\min}}{w_{\max}} = \frac{2^{m-2}}{2^{m-2}+2^{(m-1)/2}} > \frac{1}{2}.
\end{eqnarray*}

When $m \equiv \pm 3 \inmod{8}$ and $m >5$, the code $\C_D$ of Section \ref{sec-3wtcodes} satisfies that
\begin{eqnarray*}
\frac{w_{\min}}{w_{\max}} = \frac{2^{m-2}-2^{(m-1)/2}}{2^{m-2}} > \frac{1}{2}.
\end{eqnarray*}

Hence, the linear codes $\C_D$ of Section \ref{sec-3wtcodes} satisfy the condition that
$w_{\min}/w_{\max} > \frac{1}{2}$ when $m \geq 6$, and can thus be employed to obtain secret sharing schemes
with interesting access structures using the framework in \cite{YD06}. Note that binary linear codes can be employed
for secret sharing bit by bit. Hence, a secret of any size can be shared with a secret sharing scheme based on a binary
linear code. We remark that the dimension of the code $\C_D$
of this paper is small compared with its length and this makes it suitable for the application in secret sharing.

\section{Concluding remarks}

In this paper, we established a few relations among the parameters of a few subclasses of linear codes $\C_D$
(i.e., Theorem \ref{thm-mmm1}, Corollaries \ref{cor-mmm11} and \ref{cor-mmm12}, Theorem \ref{thm-nnn1},
Corollary \ref{cor-BooleanCodes}). With these relations,  a number of classes of one-weight, two-weight and
three-weight codes are derived from some known  classes of one-weight, two-weight and three-weight codes.
Instead of writing down all these codes, we documented a few classes of them as examples in this paper.
We also constructed a class of three-weight binary codes described in Theorems \ref{thm-twothree1} and
\ref{thm-twothree2}.

The codes presented in this paper are interesting, as one-weight codes, two-weight codes and three-weight codes
have applications in
secret sharing \cite{ADHK,CDY05,YD06}, authentication codes \cite{CX05}, combinatorial designs and graph
theory \cite{CG84,CK85}, and association schemes \cite{CG84}.

Every linear code over a finite field is generated by a generator matrix. Different ways of constructing the generator
matrix give different constructions of linear codes. Similarly, different ways of constructing the defining set $D$ for
the generic constriction of linear codes in this paper are different constructions of the linear codes $\C_D$. There
are a huge number of ways of constructing the defining $D$, and thus many different constructions of the codes
$\C_D$. The difficulty is the selection of $D$ so that the code $\C_D$ has good parameters. Note that the generic
construction of linear codes $\C_D$ of this paper is different from the one in \cite{CDY05}.




\begin{thebibliography}{99}
\bibitem{ADHK}  Anderson, R., Ding, C., Helleseth, T.,  Kl{\o}ve, T.: How to build robust shared control systems,
Des. Codes Cryptogr. \textbf{15}(2), 111--124 (1998)

\bibitem{CG84} Calderbank, A. R.,  Goethals, J. M.: Three-weight codes and association schemes,
Philips J. Res. \textbf{39}, 143--152 (1984)

\bibitem{CK85}
Calderbank, A. R.,  Kantor, W. M.: The geometry of two-weight codes, Bull. London Math. Soc.
\textbf{18}, 97--122 (1986)

\bibitem{CDY05} Carlet, C.,  Ding, C.,  Yuan, J.: Linear codes from perfect nonlinear
mappings and their secret sharing schemes, IEEE Trans. Inf.  Theory
\textbf{51}(6), 2089--2102 (2005)

\bibitem{Carlitz}
Carlitz, L.: Explicit evaluation of certain exponential sums, Mathematica Scandinavics.
\textbf{44}, 5--16 (1979)

\bibitem{Coult} Coulter, R. S.: On the evaluation of a class of Weil sums in characteristic $2$,
New Zerland J. of Math. \textbf{28}, 171--184 (1999)


\bibitem{Dickson}
Dickson, L. E.: Linear Groups with an Exposition of the Galois Field Theory. Dover, New York (1958)

\bibitem{Ding09} Ding, C.: A class of three-weight and four-weight codes, in: Xing C. et al. (Eds.),
Proc. of the Second International Workshop on Coding Theory and Cryptography,
Lecture Notes in Computer Science,  vol. 5557, pp. 34--42. Springer Verlag, Berlin (2009)

\bibitem{Ding15} Ding, C.: Linear codes from some 2-designs, IEEE Trans. Inf. Theory \textbf{60}(6),
3265--3275 (2015)

\bibitem{DLN} Ding, C.,  Luo, J., Niederreiter, H.: Two weight codes punctured from irreducible
cyclic codes, in: Y. Li, S. Ling, H. Niederreiter, H. Wang, C. Xing, S. Zhang (Eds.),
Proc. of the First International Workshop on Coding Theory and Cryptography,
pp. 119--124. World Scientific, Singapore (2008)

\bibitem{DN07} Ding, C.,  Niederreiter, H.: Cyclotomic linear codes of order 3,
IEEE Trans. Inf.  Theory \textbf{53}(6), 2274--2277 (2007)

\bibitem{CX05} Ding, C., Wang, X.: A coding theory construction of new systematic authentication codes,
Theoretical Computer Science \textbf{330}, 81--99 (2005)

\bibitem{DingDing1} Ding, K., Ding, C.:  Binary linear codes with three weights,
IEEE Communication Letters \textbf{18}(11), 1879--1882 (2014)

\bibitem{DingDing2} Ding, K., Ding, C.: A class of two-weight and three-weight codes and their applications in
secret sharing, to appear in IEEE Trans. Inf. Theory, arXiv:1503.06512.

\bibitem{HP} Huffman, W. C., Pless, V.: Fundamentals of Error-Correcting Codes.
Cambridge University Press, Cambridge (2003)

\bibitem{LN} Lidl, R., Niederreiter, H.: Finite Fields. Cambridge University Press,
Cambridge) 1997

\bibitem{WDX15} Wang, Q.,  Ding, K., Xue, R.:  Binary linear codes with two weights,
IEEE Communications Letters \textbf{19}, 1097--1100 (2015) .

\bibitem{YD06} Yuan, J., Ding, C.: Secret sharing schemes from three classes of linear codes,
IEEE Trans. Inf. Theory \textbf{52}(1), 206--212 (2006)

\end{thebibliography}
\end{document}